\documentclass{article}

 \usepackage[preprint]{neurips_2026}

\usepackage{amsmath,amsthm, amssymb, mathtools}
\usepackage{enumitem}
\usepackage{booktabs}
\usepackage{hyperref, comment}
\hypersetup{colorlinks=true,linkcolor=blue,citecolor=blue,urlcolor=blue}

\theoremstyle{theorem}
\newtheorem{theorem}{Theorem}
\theoremstyle{lemma}
\newtheorem{lemma}{Lemma}
\theoremstyle{proposition}
\newtheorem{proposition}{Proposition}

\theoremstyle{definition}
\newtheorem{definition}{Definition}

\newcommand{\R}{\mathbb{R}}
\newcommand{\X}{\mathcal{X}}
\newcommand{\I}{\mathbb{I}}
\newcommand{\E}{\mathbb{E}}
\newcommand{\1}{\mathbf{1}}

\begin{document}

\title{
What Limits Does Quantization Place on Dense Top-$k$ Retrieval? A Theoretical Study
}

\author{%
  Koki Okajima \\
NTT, Inc. \\
  \texttt{koki.okajima@ntt.com} \\
  \And 
    Tsukasa Yoshida \\
NTT, Inc. \\
  \texttt{tsukasa.yoshida@ntt.com} \\
}

\maketitle

\begin{abstract}
We establish conditions for embedding a corpus of $N$ documents as $d$-dimensional vectors such that every $k$-subset $S \subseteq [N]$ is realizable as a result of top-$k$ retrieval by some query vector. Recent work shows that $d = O(k)$ suffices for such embeddings to exist in $\mathbb{R}^d$, independently of $N$. We theoretically prove that this corpus-independent bound is specific to infinite precision. With $B$ bits per coordinate, perfect top-$k$ retrieval requires $Bd = \Omega(k \ln N)$; thus, at any fixed precision, the dimension must grow at least logarithmically with $N$. Specializing to a $\ell_2$-normalized $B$-bit uniform scalar quantization model, we also identify a threshold on the precision $B^{*} = O(\ln \ln N)$ below which no dimension suffices, together with two further regimes that bound the feasible $(B, d)$ pairs. Our result implies that in practical vector databases and dense retrieval systems where quantization is standard, the embedding dimension and possibly the precision must grow with the corpus size.
s 
\end{abstract}

\section{Introduction}

Dense vector embeddings form the basis of modern information retrieval,
recommendation, and retrieval-augmented generation \citep{su2023instructor, shao2025reasonir}, where a query and a corpus are
mapped into a common space $\mathbb{R}^d$, and its relevance score is decided by inner
product between these embeddings. A natural, yet fundamental question is how the embedding
dimension $d$ must scale with the corpus size $N$ if the model is to realize
every possible top-$k$ retrieval set, that is, if for every $k$-subset
$S \subseteq [N]$ there is some query vector whose top-$k$ result of the inner product with the corpus vectors is exactly $S$.
The question has gained attention as embedding models are extended
from semantic similarity to instruction-following and reasoning, regimes in
which the number of distinct top-$k$ retrieval tasks that the embeddings must
support grows combinatorially with $N$.

Recently, it has been proven that this condition is quite generous. 
\citet{weller2026limits} and \citet{wang2026r2k} have proven that for embeddings on a uniform sphere in $\mathbb{R}^d$, one only needs $d = 2k +1 = O(k)$ to embed any top-$k$ relevance relationship between the corpus and queries, which indicates that such vector representations can exist even in low dimensional spaces. 
In addition, it has been demonstrated that such embeddings are difficult to learn under standard learning procedures, highlighting the discrepancy between the existence of such embeddings and its learnability. 

However, the analysis assumes real-valued coordinates of unbounded precision. 
Deployed retrieval systems do not. In fact, 
with the advent of large-scale databases,
vector coordinates are stored in a small number of bits, with formats such as int8 \citep{jacob2018quantization,dettmers2022llmint8}, int4 \citep{Frantar2023ICLR}, fp8 \citep{micikevicius2022fp8}, and fp4 \citep{liu2023FP4} now standard, along with product quantization \citep{jegou2011pq} and related schemes.
Whether $O(k)$ bound, which is independent on the size of the corpus, 
survives this discretization is not addressed by prior work. The question is
non-trivial, since an existence argument under continuous embeddings relies on the
 geometry of $\mathbb{R}^d$, whereas quantized embeddings lie on a finite set. 

\paragraph{Contributions} In this work, we theoretically show that the corpus-independent bound indeed does not survive. Once the
embedding alphabet is finite, the existence question is governed by a
first-moment count over a discrete hypothesis space rather than by the
continuous geometry of $\mathbb{R}^d$. 
Based on this argument, our contributions are as follows: 

\begin{enumerate}[leftmargin=2em,itemsep=2pt]
\item We prove that realizing every
$k$-subset of an $N$-corpus over a $B$-bit-per-coordinate alphabet requires
$Bd = \Omega (k\ln N)$, so that the dimension must grow with $N$ at every
fixed precision and the precision must grow with $N$ at every fixed
dimension (Theorem \ref{thm:generic-bound}).  
\item Specializing to an $\ell_2$-normalized embedding model where each element is uniformly quantized, we identify a tighter bound on the necessary dimension compared to Theorem \ref{thm:generic-bound}, 
and furthermore reveal a threshold $B^* = O(\ln \ln N)$ on the number of bits $B$ below which no dimension suffices. Moreover, we further prove that there is also a ceiling for $d$ in which top-$k$ retrieval becomes impossible above it (Theorem \ref{thm:fp-bound}). 
\end{enumerate}

\section{Related Work}

\paragraph{Dimension and retrieval capacity.}
The empirical dependence of retrieval quality on embedding dimension at scale
was documented by \citet{reimers2021curse} and
\citet{yin2018dimensionality}. On the theoretical side, 
\citet{weller2026limits} establish that an embedding dimension of $d = 2k+1$ is sufficient to realize
every top-$k$ retrieval set under cosine similarity via the sign-rank
of the relevance matrix~\citep{alon1985geometrical, Forster2002Comms}, and further using a sphere
packing argument to address the case with when the scores between relevant and irrelevant documents are separated by a constant margin. \citet{wang2026r2k} reach the same 
$2k+1$ bound via shattering and VC-type arguments~\citep{mohri2018foundations}. Both
arguments take real-valued coordinates of unbounded precision as given;
neither addresses whether the same $2k+1$ dimension remains sufficient once
coordinates are stored at finite precision. 

\paragraph{The first moment method.}
The first moment method is a classical technique in probabilistic analysis~\citep{alon2016probabilistic,mezard2009information} to establish
nonexistence of some object. Here, one evaluates the upper bound on the count of said admissible objects, and observes that whenever the bound falls below one,
the count must vanish and hence the object cannot exist. The technique dates back to 
\citet{erdos1947graph}, and has since become a principal tool for locating sharp
thresholds in random discrete structures such as random
$k$-SAT~\citep{achlioptas2004threshold,friedgut1999sharp,ding2022satisfiability}, sparse random
graphs~\citep{achlioptas2005chromatic}, and random binary codes~\citep{barg2002random}. 
 We instantiate the same scheme in a new setting by counting
configurations of a finite-alphabet embedding that could shatter every $k$-subset.

\paragraph{Quantization of vector embeddings.}
Quantization is used in deployed retrieval to cut the memory and bandwidth
cost of storing large embedding indexes. 
In this work, we focus on element-wise quantization methods. \textit{Uniform scalar quantization} maps each coordinate independently to one of $2^B$ evenly-spaced levels in a fixed range, with the int4 and int8 formats \citep{jacob2018quantization, douze2024faiss, johnson2019faissgpu} being canonical choices.
Non-uniform schemes such as fp4 and fp8 instead allocate precision unevenly across the range \citep{micikevicius2022fp8}. 
Even more aggressive quantization methods which map elements to $\{0,1\}$ hash codes \citep{Salakhutdinov2009SemanticHashing, yamada2021binary} have also been shown effective.  
Other methods include vector quantization, which compresses
entire embeddings to a small codebook, with product quantization as a widely used variant \citep{jegou2011pq} to enable efficient approximate nearest neighbor search \citep{andoni2008nearoptimal, Malkov2020ANN}. 
The $B$-bit-per-coordinate alphabet used in Section~\ref{sec:fp} corresponds to
the FAISS scalar quantizer
family~\citep{douze2024faiss,johnson2019faissgpu}, while the bound
$Bd=\Omega(k\ln N)$ applies to any code.

\section{Problem Setup}\label{sec:setup}

Each query and each document of a corpus of size $N$ is embedded into a discrete
set $\X \subseteq \R^{d}$. Following the $k$-shattering formulation
of \citet{wang2026r2k}, we ask whether the corpus can be configured so that every
ground-truth set is exactly separable by inner-product score.

\begin{definition}[Existence indicator]
    For any ground-truth set $S\subseteq[N]$ with $|S|=k \leq N$, define $\I \in \{0,1\}$ as 
\[
  \I = \1\left[\,
    \exists\,\{u_i \in \X\}_{i=1}^{N},\;
    \exists\,\{v_S \in \X\}_S\;\text{s.t.}\;
    \forall S,\;
    \min_{i \in S} u_i\cdot v_S > \max_{i \notin S} u_i\cdot v_S
  \right].
\]
\end{definition}

$\I=1$ means a finite-precision configuration that serves every top-$k$ query
exists. $\I=0$ means no such configuration exists in $\X$, regardless of
training, optimization, or model capacity. Throughout,
$M \coloneqq \binom{N}{k}$ is the number of distinct ground-truth sets, and logarithms are natural.proofs of all propositions and theorems are deferred to Appendix \ref{app:proofs}.

\section{A Counting Bound for Finite Alphabets}\label{sec:counting}

A finite alphabet admits only finitely many configurations, while the number of
separation constraints grows combinatorially in $N$. A first-moment count turns
this into a quantitative bound.

\begin{proposition}[First-moment bound]\label{lem:first-moment}
\[
  \I \;\le\; \exp\bigl[\,-M\ln M + (N+M)\ln|\X|\,\bigr].
\]
\end{proposition}

Since $\I$ is a non-negative integer, $\I=0$ once the exponent is negative, that
is, once $M\ln M$ exceeds $(N+M)\ln|\X|$. Because
$M=\binom{N}{k}\ge (N/k)^{k}$ grows polynomially in $N$ of degree $k$ while
$\ln|\X|$ is a fixed budget, the budget must scale with $k\ln N$ for $\I=1$ to
remain possible. The next theorem makes the constant explicit for a
bit-quantized alphabet.

\begin{theorem}[Finite-alphabet impossibility]\label{thm:generic-bound}
Suppose $N \ge 2$ and $k \ge 2$. For any alphabet $\X$ with
$|\X| = 2^{Bd}$ and $B \ge 1$, if
\begin{equation}
\label{eq:thm1}
  d \;<\; \left(1 - \frac{2}{N}\right) \frac{k (\ln N - \ln k) }{B \ln 2}, 
\end{equation}
then $\I = 0$: no finite-precision configuration serves all top-$k$ queries.
\end{theorem}

Writing the stored budget as $Bd$ bits per vector, perfect all-subset retrieval
requires $Bd =\Omega \bigl(k\ln N\bigr)$.
The bound is symmetric in $B$ and $d$, so precision and dimension are in a trade-off relation.

\section{Uniformly Quantized Embeddings}\label{sec:fp}

While Theorem~\ref{thm:generic-bound} holds for any $2^{Bd}$-point alphabet, it is more informative to specialize to alphabets which reflect practical usage, such as the 4-bit and 8-bit scalar quantizers in FAISS~\citep{douze2024faiss,johnson2019faissgpu}.
 In particular, \texttt{QT\_$B$\_uniform} variants in FAISS are calibration-free, with a shared range and quantization width for each coordinate. 
We therefore introduce a model that quantizes each coordinate on a shared range $[-1,1]$ with a symmetric uniform mid-point grid of $2^B$ levels.
Furthermore, to reflect the standard practice of normalized embeddings, we restrict the vectors to lie within the unit $\ell_2$ ball. 
Since the ball strictly contains the unit sphere, any impossibility ($\I = 0$) proved under this model also holds when vectors are restricted to unit norm.


\begin{definition}\label{def:grid}
Let $\Lambda_B \subset [-1,1]$ be the $B$-bit uniform mid-point quantizer
alphabet, the set of $2^{B}$ points partitioning $[-1,1]$ into $2^{B}-1$ equal
subintervals of length $2^{1-B}$,
\[
  \Lambda_B \;=\; \bigl\{\,2^{-B}(2q+1)\;\big|\;
     q = -2^{B-1},\,\ldots,\,2^{B-1}-1\,\bigr\}.
\]
Define the grid $\mathcal{G}_{d,B} = \Lambda_B^{d} \subseteq [-1,1]^{d}$ and
the norm-bounded subset
\[
  \mathcal{S}_{d,B} \;=\; \bigl\{\,x \in \mathcal{G}_{d,B} \,\big|\, \|x\|_2 \le 1\,\bigr\}.
\]
\end{definition}

$\mathcal{S}_{d,B}$ is the set of vectors a normalized, $B$-bit-quantized system
can store. 
The next theorem provide the conditions in which perfect top-$k$ retrieval becomes impossible, which we prove in the Appendix by evaluating $|\mathcal{S}_{d,B}|$ and applying Proposition~\ref{thm:generic-bound}. 

\begin{theorem}[Finite-precision impossibility]\label{thm:fp-bound}
Suppose $N \ge 7$ and $k \ge 2$, and let
$\{u_1,\ldots,u_N\} \subseteq \mathcal{S}_{d,B}$. Define $B^*$ from 
\begin{equation} \label{eq:Bstar}
    B^* \ln 4 = \ln \ln M - \ln \pi + \ln \frac{M}{N+M}. 
\end{equation}
If any one of the following
holds, then $\I = 0$; i.e. perfect top-$k$ retrieval is impossible. 
\begin{enumerate}[label=\textup{(\roman*)}, leftmargin=2.4em,itemsep=2pt]
\item \emph{(Sub-critical precision)} $B \leq B^{*}$. 
\item \emph{(Precision-limited dimension)} for $B = (1+\delta)B^{*}$,
  $\delta > 0$,
  \begin{equation}
  \label{eq:thm2_2}
          d \;<\; \frac{2k}{e} \left(1-\frac{2}{N}\right) \frac{ \ln N - \ln k }{1+ \delta B^* \ln 4}. 
  \end{equation}
\item \emph{(High dimension)} for $B = (1+\delta)B^{*}$, $\delta \ge 0$,
  \begin{equation}\label{eq:thm2_3}
    d \;>\; 2e\,[k\ln N + k - k \ln k]^{1+\delta}.  
  \end{equation}
\end{enumerate}
\end{theorem}
The implication of these results will be addressed in the next section. 
Note that since $M \geq (N/k)^k$, $B^*$ is $O(\ln \ln N)$ for $k = O(1)$. 

\section{Relation to prior work}\label{sec:discussion}

\paragraph{Reconciliation with the $O(k)$ bounds.}
The $O(k)$ embedding results of Wang et al.~\cite{wang2026r2k} and the
sign-rank construction of Alon et al.~\cite{alon1985geometrical} invoked in
Weller et al.~\cite{weller2026limits} both assume real-valued coordinates of
unbounded precision. In the notation of this paper, this is the limit
$B \to \infty$. 
However, in this limit the inequalities \eqref{eq:thm1} and \eqref{eq:thm2_2} holds trivially for any $d \geq 1$. The bound
is therefore vacuous when precision is unconstrained, and it imposes no
restriction on the corpus-independent dimension $d = 2k+1$ obtained by
the previous constructions assuming infinite precision. The two results are thus not in conflict, but they describe the distinct behaviors under different assumptions on $B$. 

\paragraph{Practical implications.}
Theorem~\ref{thm:generic-bound} establishes that the binding quantity over
any finite alphabet is the product $Bd$, the bit budget per document
embedding. Since the bound is symmetric in $B$ and $d$, halving the
precision at fixed corpus size doubles the minimum admissible dimension,
and halving the dimension at fixed corpus size doubles the minimum
admissible precision. 
Theorem~\ref{thm:fp-bound} sharpens this picture for the $\ell_2$-normalized scalar quantizers. 
Apart from condition (i) which tightens the bound for the general case \eqref{eq:thm1}, its principal consequence is the hard precision floor $B^{*} = O(\ln \ln N)$ in (ii), below which no dimension suffices for
perfect top-$k$ retrieval. This implies that under an overwhelmingly large corpus, one must also provide a sufficiently fine quantization \textit{irrespective of how large the embedding dimension is.} This condition is tighter than the generic bound in Theorem~\ref{thm:generic-bound}, where the necessary condition can be satisfied by increasing $d$ to compensate a small $B$. 
Regime (iii) establishes a ceiling on dimension at fixed precision. As $d$ grows, the unit ball's volume shrinks relative to the volume of the cube $[-1,1]^d$, so the norm-bounded grid $\mathcal{S}_{d,B}$ admits too few vectors to realize all $\binom{N}{k}$ top-$k$ patterns even though the underlying grid spacing $2^{1-B}$ is unchanged. 

\section{Numerical Illustration}\label{sec:numerics}
Figure~\ref{fig:thm1} illustrates the necessary lower bound for $d$ against different $N$ for $B = 4$ under Theorem~\ref{thm:generic-bound}, where perfect top-$k$ retrieval fails for any $d$ below the corresponding curve. We can see that for moderately large $k$, such as $k = 128$, the necessary dimension exceeds values used in standard dense retrievers (such as $d = 384$ \citep{Wang2020MiniLM}) for large corpora with $N$ larger than $10^6$. 

The left panel of Figure~\ref{fig:thm2} illustrates the necessary lower and upper bounds under Theorem~\ref{thm:fp-bound} for $d$ under the same conditions. In particular, for $k = 128$, $B = 4$ cannot exceed $B^*$ for $N$ larger than $N \approx 2.5 \cdot 10^4$, indicating a stronger restriction on $d$ compared to Theorem 1. The same cutoff is also observed for $k = 64$, albeit at a larger $N \approx 7 \cdot 10^6$. However, as we see in the right panel of Figure~\ref{fig:thm2}, these cutoffs can be avoided fairly easily by slightly increasing $B$, as $B^*$ only grows as $O(\ln \ln N).$ 
Finally, the upper bound on $d$ given in condition (iii) of Theorem~\ref{thm:fp-bound} only have mild consequences, as standard dense retrievers rarely exceed $d = 4096$ \citep{Wang2024ImprovingText}, while the upper ceiling in the left panel of Figure~\ref{fig:thm2} exceeds 4300. 

\begin{figure}
    \centering
    \includegraphics[width=0.7\linewidth]{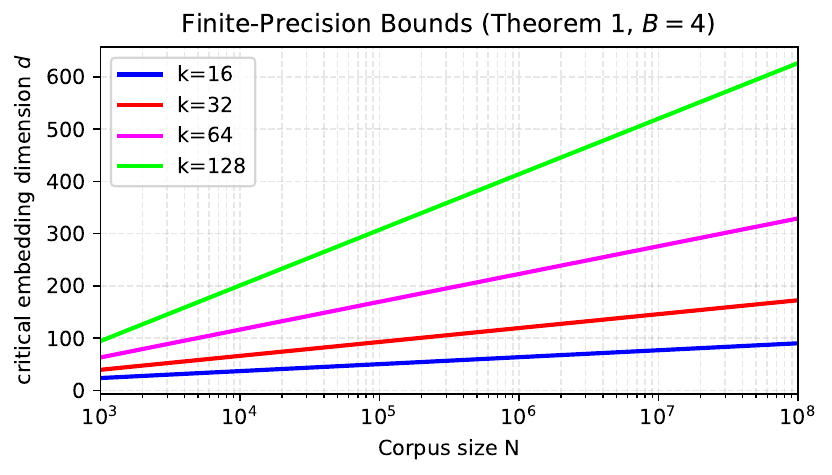}
    \caption{Critical embedding dimension (right-hand side of \eqref{eq:thm1}) against $N$ for $B = 4$ and different values of $k$.}
    \label{fig:thm1}
\end{figure}
\begin{figure}
    \centering
    \includegraphics[width=1.0\linewidth]{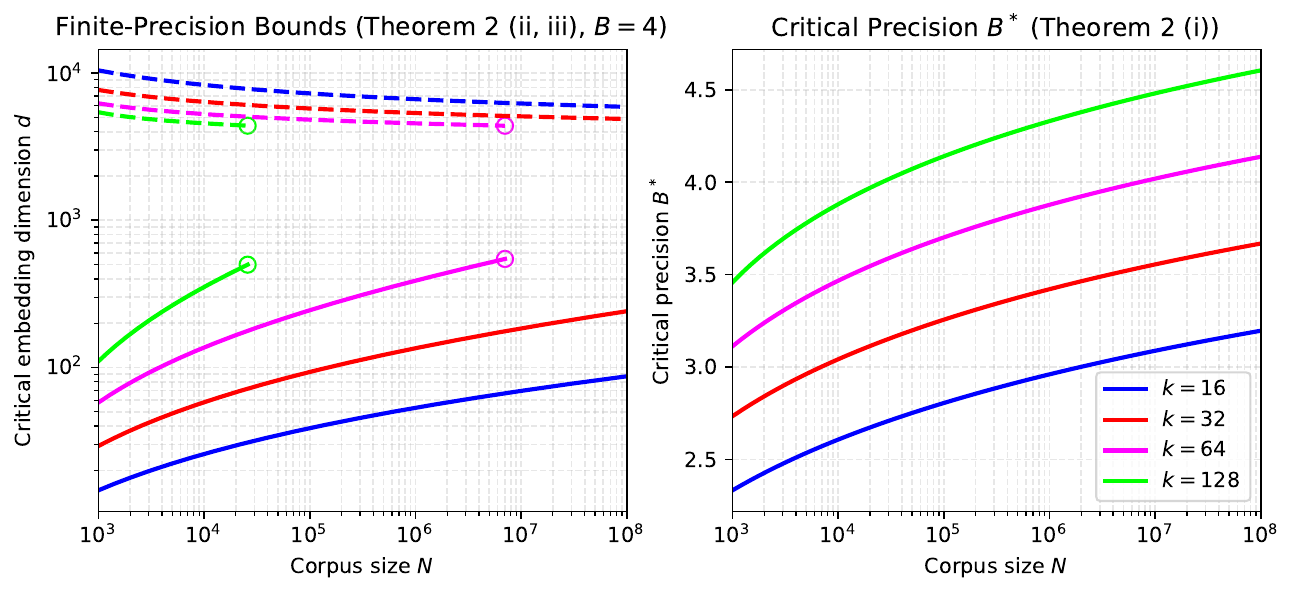}
    \caption{(Left:) Critical embedding dimensions given by the right-hand side of \eqref{eq:thm2_2} (solid line) and \eqref{eq:thm2_3} (broken line) against $N$ for $B = 4$. Curves terminate when $B^*$ exceeds $B$, in which any embedding dimension yields $\I = 0$.  (Right:) Critical precision $B^*$, defined in \eqref{eq:Bstar}, against $N$. }
    \label{fig:thm2}
\end{figure}

\section{Limitations and Future Work}

The analysis addresses perfect top-$k$ retrieval, requiring every
$k$-subset of $[N]$ to be realizable as the result of some query. An
approximate formulation, in which a controlled fraction of queries may
yield an incorrect set, would better reflect the metrics used to evaluate
deployed systems and would likely relax regime (i) of
Theorem~\ref{thm:fp-bound}. Establishing such bounds is left to
future work.

The quantization model considered in Definition \ref{def:grid} is uniform mid-point scalar quantization with a shared range $[-1,1]$. This does not fully reflect the properties of some basic quantization schemes, as deployed systems also employ learned, non-uniform codebooks and product quantization schemes that decompose the embedding space before discretization. For such schemes, Theorem~\ref{thm:generic-bound} continues to apply as the relevant
fallback, whereas Theorem~\ref{thm:fp-bound} does not extend directly.
Establishing analogous precision-specific bounds for learned codebooks
and product quantization is also an interesting direction of research.

Finally, the results are impossibility statements; we do not exhibit explicit finite-precision configurations attaining either
Theorem~\ref{thm:generic-bound} or Theorem~\ref{thm:fp-bound}, and the
precision floor $B^{*} = O(\log \log N)$ are not claimed tight. While not attempted in this work, this could possibly be addressed using lower-bounding techniques such as the second moment method \citep{mezard2009information}. 

\section{Conclusion}

This paper studied the existence of finite-precision dense embeddings
that realize every top-$k$ retrieval over an $N$-document corpus.
A first-moment counting argument over the discrete configuration space
established $Bd = \Omega(k \ln N)$ for any $2^{Bd}$-point alphabet, and
the same scheme applied to $\ell_2$-normalized uniform scalar quantization
yielded a precision floor of order $\ln \ln N$, together with lower and
upper bounds on the necessary embedding dimension.
These results separate the unbounded-precision and finite-precision
regimes and quantify how, in the latter, the admissible embedding
geometry depends on the corpus size.

\appendix
\section{Proofs}\label{app:proofs}

\subsection*{Auxiliary Lemmas}
In this subsection we state and prove two lemmas which are crucial to the main proof of Theorems \ref{thm:generic-bound} and \ref{thm:fp-bound}. 

\begin{lemma}[Cardinality of $\mathcal{S}_{d,B}$]\label{lem:cardinality}
\[
  |\mathcal{S}_{d,B}|
  \;\le\; \exp\left[\,d\left(B \ln 2 - \frac12\ln d
        + \frac12 \ln(2\pi e)\right)\right].
\]
\end{lemma}
\begin{proof}
    For any $t>0$, $\1[\|x\|_2^2\le 1]\le e^{t(1-\|x\|_2^2)}$. Summing over
$\mathcal{G}_{d,B}$ and factorizing,
\begin{equation}\label{eq:tmp1}
  |\mathcal{S}_{d,B}|
  \le e^{t}\Bigl(\sum_{x\in\Lambda_B} e^{-t x^2}\Bigr)^{d}.    
\end{equation}
By Definition~\ref{def:grid}, with $q$ ranging over odd integers,
\begin{align*}
  \sum_{x\in\Lambda_B} e^{-t x^2}
   = \sum_{q=-2^{B}+1,\ q:\text{odd}}^{2^{B}-1}
      e^{-\frac{t}{4^{B}} q^2}
  \le 2\sum_{q=1}^{\infty} e^{-\frac{t}{4^{B}} q^2}  <  2\int_{0}^{\infty} e^{-\frac{t}{4^{B}}x^2}\,dx 
  = 2^{B}\sqrt{\pi/t},    
\end{align*}
where the strict inequality holds since the summand is decreasing. Therefore, from \eqref{eq:tmp1}, 
$
  |\mathcal{S}_{d,B}|
  \le \exp\left(t + dB\ln 2 - \frac{d}{2}\ln t + \frac{d}{2}\ln\pi\right),
$
and substituting the optimizing choice $t=d/2$ yields the claim.
\end{proof}

\begin{lemma}[bound on the Lambert-$W$ function]\label{lem:W-neg}
Let $W_{-1}$ be the lower branch of the Lambert $W$ function. For any $u>0$,
\[
  W_{-1}\left(-e^{-u-1}\right)
  > -(u+2) - \ln(u+1).
\]
\end{lemma}
\begin{proof}
Let $t=-W_{-1}(-e^{-u-1})>1$, so $t-\ln t=u+1$, hence $t>u+1$ and
$t=\ln t+u+1>u+1+\ln(u+1)$. Put $\rho=t-(u+1)-\ln(u+1)>0$. Then
\[
  \rho=\ln t-\ln(u+1)
  =\ln\Bigl(1+\frac{\rho+\ln(u+1)}{u+1}\Bigr)
  <\frac{\rho+\ln(u+1)}{u+1},
\]
so $\rho<\frac1u\ln(u+1)$, giving
$t<u+1+\frac{u+1}{u}\ln(u+1)$. Finally using $u^{-1}\ln(u+1) \leq 1$ offers $t<u+2+\ln(u+1),$ resulting in the statement. 
\end{proof}

\subsection*{Proof of Proposition~\ref{lem:first-moment}}
We bound $\I$ by the count of admissible configurations and apply the first
moment method:
\begin{align*}
  \I
  &\le \#\left[
      \{u_i \in \X\}_{i=1}^{N},\,\{v_S \in \X\}_S
      \;\Big|\;
      \forall S,\; \min_{i \in S} u_i\cdot v_S > \max_{i \notin S} u_i\cdot v_S
    \right] \\
  &= |\X|^{N}\,\E_u\Biggl[\,
      \prod_S \sum_{v_S \in \X} \,
      \1\Bigl[\min_{i \in S} u_i\cdot v > \max_{i \notin S} u_i\cdot v\Bigr]
    \Biggr] ,
\end{align*}
where $\E_u$ is the average over all configurations of $u\in \X$. 
Using the AM-GM inequality on the product over $S$, i.e. $\prod_S a_S \leq (M^{-1} \sum_S a_S)^M$, offers
\begin{equation*}\I \le |\X|^{N+M}\,\E_u \Biggl[
      M^{-1} \E_{v} \sum_S 
      \1\Bigl[\min_{i \in S} u_i\cdot v > \max_{i \notin S} u_i\cdot v\Bigr] \Biggr]^{M}.
\end{equation*}
For any fixed configuration of $\{u_i\}$ and $v$, at
most one $S$ can satisfy the strict separation, so $\sum_S \1[\cdots] \leq 1$, yielding
$
  \I \le |\X|^{N+M}\,\E_u\left[M^{-1}\right]^{M}
  = \exp\bigl[\,-M \ln M + (N+M)\ln|\X|\,\bigr]. \qedhere
$

\subsection*{Proof of Theorem~\ref{thm:generic-bound}}
With $|\X|=2^{Bd}$, Lemma~\ref{lem:first-moment} offers
\begin{align*}
  \I
  &\le \exp\left[(N+M)\,B \ln 2\,
      \Bigl(d - \frac{M}{N+M}\,\frac{\ln M}{B \ln 2}\Bigr)\right] .
\end{align*}
Since $M\ge\binom{N}{2}=\frac12 N(N-1), 1-\frac{N}{N+M} \geq 1 - \frac{2}{N+1} > 1 - \frac{2}{N}$ for $N \geq 2$. 
Additionally using $M \ge (N/k)^{k}$ offers 
\[
  \I \le \exp\left[(N+M)\,B\ln 2\,
      \Bigl(d - \frac{k(\ln N - \ln k)}{B\ln 2}\Bigl(1 - \frac{2}{N}\Bigr) \Bigr)\right].
\]
$\I$ is a non-negative integer, so it is $0$ once the exponent is negative,
which is the stated condition. \qed

\subsection*{Proof of Theorem~\ref{thm:fp-bound}}
Combining Proposition~\ref{lem:first-moment} with Lemma~\ref{lem:cardinality},
\[
  \frac{1}{N+M}\ln\I
  \le \pi e\,4^{B}\left(
      -\frac{M}{N+M}\cdot\frac{\ln M}{\pi e\,4^{B}}
      - \frac{d}{2\pi e\,4^{B}}\ln\frac{d}{2\pi e\,4^{B}}\right).
\]
The minimum over $d$ of $\frac{d}{2\pi e\,4^{B}}\ln\frac{d}{2\pi e\,4^{B}}$ is
$-1/e$, so the right side is uniformly negative in $d$, forcing $\I=0$ whenever the following holds:
\[
  \frac{M}{N+M}\cdot\frac{\ln M}{\pi e\,4^{B}} \ge \frac1e,
  \ \text{i.e.}\
  B \le \frac{1}{\ln 4}\biggl[\ln\ln M-\ln\pi+\ln\frac{M}{N+M}\biggr]
  \equiv B^{*},
\]
 proving (i). Note that $N>7$ and $k>2$ guarantees $B^\star > 0$. For $B>B^{*}$ set
\[
  w=\ln \frac{d}{2\pi e\,4^{B}} \quad \text{and} \quad  
  y=e^{-1} 4^{B^{*}-B},
\]
so \[\frac{1}{\pi e\,4^{B}}\cdot\frac{1}{N+M}\ln\I\le -y-w e^{w}.\] This is
positive (the bound fails to certify $\I=0$) exactly when
$2\pi 4^{B}e^{W_{-1}(-y)+1}\le d\le 2\pi 4^{B}e^{W_{0}(-y)+1}$. Outside this
interval $\I=0$. With $u=\ln4^{B-B^{*}}=\delta B^{*}\ln4$,
Lemma~\ref{lem:W-neg} and the lower bound on $B^*$ via $\frac{M}{N+M} \geq 1-\frac{2}{N}$ give 
\[
  2\pi 4^{B}e^{W_{-1}(-y)+1}
  > \frac{2\pi}{e}\cdot\frac{4^{B^{*}}}{u+1}
  \ge \frac{2k}{e} \left(1-\frac{2}{N}\right) \frac{ \ln N - \ln k }{1+ \delta B^* \ln 4}
       ,
\] which is regime (ii). For the upper threshold,
using $W_{0}(-y)\le 0$ for $y\ge 0$, and with
$\ln M\le k (\ln N + k - \ln k )$,
\[
  2\pi 4^{B}e^{W_{0}(-y)+1}\le 2\pi e\,4^{(1+\delta)B^*}
  \le 2e\,[k\ln N + k - k \ln k]^{1+\delta}
\]
for $\delta\ge0$, which is regime (iii). \qed

\bibliography{ref}

@inproceedings{weller2026limits,
  author    = {Orion Weller and Michael Boratko and Iftekhar Naim and Jinhyuk Lee},
  title     = {On the Theoretical Limitations of Embedding-Based Retrieval},
  booktitle = {International Conference on Learning Representations (ICLR)},
  year      = {2026},
}

@article{wang2026r2k,
  author  = {Zihao Wang and Hang Yin and Lihui Liu and Hanghang Tong and Yangqiu Song and Ginny Wong and Simon See},
  title   = {$\mathbb{R}^{2k}$ is Theoretically Large Enough for Embedding-based Top-$k$ Retrieval},
  journal = {arXiv preprint arXiv:2601.20844},
  year    = {2026},
}

@inproceedings{alon1985geometrical,
  author    = {Noga Alon and Peter Frankl and Vojt{\v{e}}ch R{\"o}dl},
  title     = {Geometrical realization of set systems and probabilistic communication complexity},
  booktitle = {26th Annual Symposium on Foundations of Computer Science (FOCS)},
  pages     = {277--280},
  year      = {1985},
  publisher = {IEEE Computer Society}
}

@book{mohri2018foundations,
  author    = {Mehryar Mohri and Afshin Rostamizadeh and Ameet Talwalkar},
  title     = {Foundations of Machine Learning},
  edition   = {2nd},
  publisher = {MIT Press},
  year      = {2018}
}

@inproceedings{reimers2021curse,
  author    = {Nils Reimers and Iryna Gurevych},
  title     = {The Curse of Dense Low-Dimensional Information Retrieval for Large Index Sizes},
  booktitle = {Proceedings of the 59th Annual Meeting of the Association for Computational Linguistics and the 11th International Joint Conference on Natural Language Processing (Volume 2: Short Papers)},
  pages     = {605--611},
  year      = {2021},
  publisher = {Association for Computational Linguistics}
}

@inproceedings{yin2018dimensionality,
  author    = {Zi Yin and Yuanyuan Shen},
  title     = {On the Dimensionality of Word Embedding},
  booktitle = {Advances in Neural Information Processing Systems (NeurIPS)},
   publisher = {Curran Associates, Inc.},
  volume    = {31},
  year      = {2018}
}

@article{jegou2011pq,
  author  = {Herv{\'e} J{\'e}gou and Matthijs Douze and Cordelia Schmid},
  title   = {Product Quantization for Nearest Neighbor Search},
  journal = {IEEE Transactions on Pattern Analysis and Machine Intelligence},
  volume  = {33},
  number  = {1},
  pages   = {117--128},
  year    = {2011},
}

@article{andoni2008nearoptimal,
  author  = {Alexandr Andoni and Piotr Indyk},
  title   = {Near-Optimal Hashing Algorithms for Approximate Nearest Neighbor in High Dimensions},
  journal = {Communications of the ACM},
  volume  = {51},
  number  = {1},
  pages   = {117--122},
  year    = {2008},
}

@article{douze2024faiss,
  author  = {Matthijs Douze and Alexandr Guzhva and Chengqi Deng and Jeff Johnson and Gergely Szilvasy and Pierre-Emmanuel Mazar{\'e} and Maria Lomeli and Lucas Hosseini and Herv{\'e} J{\'e}gou},
  title   = {The {Faiss} library},
  journal = {arXiv preprint arXiv:2401.08281},
  year    = {2024}
}

@article{johnson2019faissgpu,
  author  = {Jeff Johnson and Matthijs Douze and Herv{\'e} J{\'e}gou},
  title   = {Billion-Scale Similarity Search with {GPUs}},
  journal = {IEEE Transactions on Big Data},
  volume  = {7},
  number  = {3},
  pages   = {535--547},
  year    = {2021},
}

@inproceedings{jacob2018quantization,
  author    = {Benoit Jacob and Skirmantas Kligys and Bo Chen and Menglong Zhu and Matthew Tang and Andrew Howard and Hartwig Adam and Dmitry Kalenichenko},
  title     = {Quantization and Training of Neural Networks for Efficient Integer-Arithmetic-Only Inference},
  booktitle = {Proceedings of the IEEE Conference on Computer Vision and Pattern Recognition (CVPR)},
  pages     = {2704--2713},
  year      = {2018}
}

@inproceedings{dettmers2022llmint8,
  author    = {Tim Dettmers and Mike Lewis and Younes Belkada and Luke Zettlemoyer},
  title     = {{LLM.int8()}: 8-bit Matrix Multiplication for Transformers at Scale},
  booktitle = {Advances in Neural Information Processing Systems (NeurIPS)},
   publisher = {Curran Associates, Inc.},
  volume    = {35},
  year      = {2022},
}

@book{mezard2009information,
  author    = {Marc M{\'e}zard and Andrea Montanari},
  title     = {Information, Physics, and Computation},
  series    = {Oxford Graduate Texts},
  publisher = {Oxford University Press},
  year      = {2009}
}

@book{alon2016probabilistic,
  author    = {Noga Alon and Joel H. Spencer},
  title     = {The Probabilistic Method},
  edition   = {4th},
  publisher = {Wiley},
  year      = {2016}
}

@article{erdos1947graph,
  author  = {Paul Erd{\H o}s},
  title   = {Some remarks on the theory of graphs},
  journal = {Bulletin of the American Mathematical Society},
  volume  = {53},
  number  = {4},
  pages   = {292--294},
  year    = {1947}
}

@article{achlioptas2004threshold,
  author  = {Dimitris Achlioptas and Yuval Peres},
  title   = {The threshold for random $k$-{SAT} is $2^{k}\log 2 - O(k)$},
  journal = {Journal of the American Mathematical Society},
  volume  = {17},
  number  = {4},
  pages   = {947--973},
  year    = {2004}
}

@article{achlioptas2005chromatic,
  author  = {Dimitris Achlioptas and Assaf Naor},
  title   = {The two possible values of the chromatic number of a random graph},
  journal = {Annals of Mathematics},
  volume  = {162},
  number  = {3},
  pages   = {1335--1351},
  year    = {2005}
}

@article{friedgut1999sharp,
  author  = {Ehud Friedgut},
  title   = {Sharp thresholds of graph properties, and the $k$-{SAT} problem},
  journal = {Journal of the American Mathematical Society},
  volume  = {12},
  number  = {4},
  pages   = {1017--1054},
  year    = {1999},
}

@article{ding2022satisfiability,
  author  = {Jian Ding and Allan Sly and Nike Sun},
  title   = {Proof of the satisfiability conjecture for large $k$},
  journal = {Annals of Mathematics},
  volume  = {196},
  number  = {1},
  pages   = {1--388},
  year    = {2022},
}

@article{barg2002random,
  author  = {Alexander Barg and G. David Forney},
  title   = {Random codes: minimum distances and error exponents},
  journal = {IEEE Transactions on Information Theory},
  volume  = {48},
  number  = {9},
  pages   = {2568--2573},
  year    = {2002},
}

@article{micikevicius2022fp8,
  author  = {Paulius Micikevicius and Dusan Stosic and Neil Burgess and Marius Cornea and Pradeep Dubey and Richard Grisenthwaite and Sangwon Ha and Alexander Heinecke and Patrick Judd and John Kamalu and Naveen Mellempudi and Stuart Oberman and Mohammad Shoeybi and Michael Siu and Hao Wu},
  title   = {{FP8} Formats for Deep Learning},
  journal = {arXiv preprint arXiv:2209.05433},
  year    = {2022},
}

@inproceedings{yamada2021binary,
    title = "Efficient Passage Retrieval with Hashing for Open-domain Question Answering",
    author = "Yamada, Ikuya  and
      Asai, Akari  and
      Hajishirzi, Hannaneh",
    booktitle = "Proceedings of the 59th Annual Meeting of the Association for Computational Linguistics and the 11th International Joint Conference on Natural Language Processing (Volume 2: Short Papers)",
    year = "2021",
    publisher = "Association for Computational Linguistics",
    pages = "979--986",
}

@inproceedings{su2023instructor,
    title = "One Embedder, Any Task: Instruction-Finetuned Text Embeddings",
    author = "Su, Hongjin  and
      Shi, Weijia  and
      Kasai, Jungo  and
      Wang, Yizhong  and
      Hu, Yushi  and
      Ostendorf, Mari  and
      Yih, Wen-tau  and
      Smith, Noah A.  and
      Zettlemoyer, Luke  and
      Yu, Tao",
    booktitle = "Findings of the Association for Computational Linguistics: ACL 2023",
    month = jul,
    year = "2023",
    publisher = "Association for Computational Linguistics",
    pages = "1102--1121"
}

@inproceedings{shao2025reasonir,
  title     = {ReasonIR: Training Retrievers for Reasoning Tasks},
  author    = {Shao, Rulin and Qiao, Rui and Kishore, Varsha and Muennighoff, Niklas and Lin, Xi Victoria and Rus, Daniela and Low, Bryan Kian Hsiang and Min, Sewon and Yih, Wen-tau and Koh, Pang Wei and Zettlemoyer, Luke},
  booktitle = {Conference on Language Modeling (COLM)},
  year      = {2025}
}

@inproceedings{Wang2020MiniLM,
 author = {Wang, Wenhui and Wei, Furu and Dong, Li and Bao, Hangbo and Yang, Nan and Zhou, Ming},
 booktitle = {Advances in Neural Information Processing Systems},
 pages = {5776--5788},
 publisher = {Curran Associates, Inc.},
 title = {MiniLM: Deep Self-Attention Distillation for Task-Agnostic Compression of Pre-Trained Transformers},
 volume = {33},
 year = {2020}
}

@inproceedings{Wang2024ImprovingText,
    title = "Improving Text Embeddings with Large Language Models",
    author = "Wang, Liang  and
      Yang, Nan  and
      Huang, Xiaolong  and
      Yang, Linjun  and
      Majumder, Rangan  and
      Wei, Furu",
    booktitle = "Proceedings of the 62nd Annual Meeting of the Association for Computational Linguistics (Volume 1: Long Papers)",
    month = aug,
    year = "2024",
    publisher = "Association for Computational Linguistics"
}

@ARTICLE{Malkov2020ANN,
  author={Malkov, Yu A. and Yashunin, D. A.},
  journal={IEEE Transactions on Pattern Analysis and Machine Intelligence}, 
  title={Efficient and Robust Approximate Nearest Neighbor Search Using Hierarchical Navigable Small World Graphs}, 
  year={2020},
  volume={42},
  number={4},
  pages={824-836},
}

@article{Salakhutdinov2009SemanticHashing,
title = {Semantic hashing},
journal = {International Journal of Approximate Reasoning},
volume = {50},
number = {7},
pages = {969-978},
year = {2009},
note = {Special Section on Graphical Models and Information Retrieval},
issn = {0888-613X},
author = {Ruslan Salakhutdinov and Geoffrey Hinton}
}

@article{Forster2002Comms,
title = {A linear lower bound on the unbounded error probabilistic communication complexity},
journal = {Journal of Computer and System Sciences},
volume = {65},
number = {4},
pages = {612-625},
year = {2002},
note = {Special Issue on Complexity 2001},
author = {Jürgen Forster}
}

@inproceedings{Frantar2023ICLR,
  author    = {Elias Frantar and Saleh Ashkboos and Torsten Hoefler and Dan Alistarh},
  title     = {GPTQ: Accurate Post-Training Quantization for Generative Pre-trained Transformers},
  booktitle = {International Conference on Learning Representations (ICLR)},
  year      = {2023},
}

@inproceedings{liu2023FP4,
    title = "{LLM}-{FP}4: 4-Bit Floating-Point Quantized Transformers",
    author = "Liu, Shih-yang  and
      Liu, Zechun  and
      Huang, Xijie  and
      Dong, Pingcheng  and
      Cheng, Kwang-Ting",
    booktitle = "Proceedings of the 2023 Conference on Empirical Methods in Natural Language Processing",
    month = dec,
    year = "2023",
    publisher = "Association for Computational Linguistics",
    pages = "592--605"
}

\end{document}